\documentclass[twosided,11pt]{article}

\newfont{\eaddfnt}{phvr8t at 12pt}
\def\email#1{{{\eaddfnt{\par #1}}}}       
\newfont{\affaddr}{phvr8t at 10pt}
\newcount\aucount
\newcount\originalaucount
\newdimen\auwidth
\newdimen\auskip
\newcount\auskipcount
\newdimen\auskip
\global\auskip=1pc
\newdimen\allauboxes
\allauboxes=\auwidth
\newtoks\addauthors
\newcount\addauflag
\global\addauflag=0 
\newtoks\subtitletext

\gdef\numberofauthors#1{\global\aucount=#1
\ifnum\aucount>3\global\originalaucount=\aucount \global\aucount=3\fi 
\global\auskipcount=\aucount\global\advance\auskipcount by 1
\global\multiply\auskipcount by 2
\global\multiply\auskip by \auskipcount
\global\advance\auwidth by -\auskip
\global\divide\auwidth by \aucount}

\def\alignauthor{
\end{tabular}%
  \begin{tabular}[t]{p{\auwidth}}\centering}%

\def\keywords{
\section*{Keywords}
}

\usepackage{amsthm}
\usepackage{bbm}
\usepackage{float}
\usepackage{titlesec}
\usepackage{fontenc}
\usepackage[utf8]{inputenc}
\usepackage{mathtools}
\usepackage{subfigure}

\usepackage{float}
\usepackage[utf8]{inputenc}	
\usepackage[sort&compress, numbers]{natbib} 
\usepackage[lined,boxed,linesnumbered]{algorithm2e}

\newtheoremstyle{mytheorem}
{2pt}
{2pt}
{\itshape}
{}
{\bfseries}
{.}
{.5em}
{}
\theoremstyle{mytheorem}
\newtheorem{theorem}{Theorem}[section]

\newtheorem{lemma}{Lemma}[section]
\newtheorem{definition}{Definition}[section]

\newtheoremstyle{mydefinition}
{2pt}
{2pt}
{}
{}
{\bfseries}
{.}
{.5em}
{}
\theoremstyle{mydefinition}

\newtheorem{open problem}{Open Problem}

\newtheorem{remark}{Remark}[section]

\usepackage{amsmath,amssymb,latexsym,xspace,float,multirow}

\usepackage{multirow}
 \usepackage{graphics}
 \usepackage{graphicx}
 \usepackage{pdfpages}
 \usepackage{pst-plot}
 \usepackage{pstricks, pst-node, pst-circ}
 \usepackage{multido,pst-slpe}
  \SpecialCoor
  \usepackage{caption}
   \usepackage{moredefs}

\usepackage{graphics, color, epsfig,verbatim,alltt}

\usepackage{subfigure,url,booktabs}
\usepackage[breaklinks=true]{hyperref}
\usepackage{paralist}

\titlespacing*{\section}
{0pt}{.5ex}{.5ex}
\titlespacing*{\subsection}
{0pt}{.5ex}{.5ex}

\usepackage{graphics, color, epsfig, amssymb,amsmath,verbatim,alltt}
\title{Distributed Multi-Depot Routing without Communications}

\numberofauthors{3}
\author{
\alignauthor
Dawsen Hwang\\
       \affaddr{Massachusetts Institute of Technology}\\
       \affaddr{77 Massachusetts Avenue}\\
       \affaddr{Cambridge, Massachusetts}\\
       \email{dawsen@mit.edu}
\alignauthor
Patrick Jaillet\\
       \affaddr{Massachusetts Institute of Technology}\\
       \affaddr{77 Massachusetts Avenue}\\
       \affaddr{Cambridge, Massachusetts}\\
       \email{jaillet@mit.edu}
\alignauthor Zhengyuan Zhou \\
       \affaddr{Stanford University}\\
       \affaddr{	450 Serra Mall}\\
       \affaddr{Stanford, California}\\
       \email{zyzhou@stanford.edu}
}

\begin{document}
\setcounter{page}{0}
\maketitle
\begin{abstract}
We consider and formulate a class of distributed
multi-depot routing problems, where servers are to visit a set of requests, with the aim of minimizing the total distance travelled by all servers. These problems fall into two categories: distributed offline routing problems where all the requests that need to be visited are known from the start; distributed online routing problems where the requests come to be known incrementally. A critical and novel feature of our formulations is that communications are not allowed among the servers, hence posing an interesting and challenging question: what performance can be achieved in comparison to the best possible solution obtained from an omniscience planner with perfect communication capabilities?
The worst-case (over all possible request-set instances) performance metrics are given by the approximation ratio (offline case) and the competitive ratio (online case). 

Our first result indicates that the online and offline problems are effectively equivalent: for the same request-set instance, the approximation ratio and the competitive ratio differ by at most an additive factor of $2$, irrespective of the release dates in the online case.
Therefore, we can restrict our attention to the offline problem.
For the offline problem, we show that the approximation ratio given by the Voronoi partition is $m$ (the number of servers). For two classes of depot configurations, when the depots form a line and when the ratios between the distances of pairs of depots are upper bounded by a sublinear function $f(m)$ (i.e.,\ $f(m) =  \mathopen{o}(m)$), we give partition schemes with sublinear approximation ratios $O(\log m)$ and $\Theta(f(m))$ respectively.  We also discuss several interesting open problems in our formulations: in particular, how our initial results (on the two deliberately chosen classes of depots) shape our conjecture on the open problems.  
\end{abstract}

\newpage
\section{Introduction}
With the advance of technology, it is now possible to deploy a fleet of servers (e.g., UAVs, robots) to visit requests located in a surrounding territory. 
The problem can be modeled as the classical (uncapacitated) multi-depot vehicle routing problem, in which a team of servers are to collectively visit a set of requests located in an ambient metric space so as to minimize the total travelled distance.
Clearly, an optimal solution to this problem will be achieved by centralized planners that know all the requests, or equivalently, by allowing full communications among the servers.
However, such a centralized approach suffers from two drawbacks:
Practically, deploying a full-communication scheme among all servers is often overly costly and unreasonble, if not, infeasible.
Theoretically, even with all the information, 
computing the optimal assignment (i.e., which server visits which subset of requests) is intractable. 

As such, typical existing solutions to this problem are heuristic algorithms that involve some local communications among subsets of servers for relaying requests to each other~\cite{RefWorks:358,RefWorks:517,RefWorks:518}.
While empirical demonstrations indicate the effectiveness of the heuristic algorithms under certain request-set configurations, no worst-case performance guarantees have been presented. 

We take an ambitious step back and ask the following question: what if we simply disallow any communications among the servers?
We can achieve this by a static partition of the underlying metric space, determined once for all based solely on the depot locations: thereafter, independent of the request-set instance, each server is only responsible for requests in its prescribed region.
In this paper, we formulate a class of distributed routing problems and study such static partitions.

\subsection{Our Contributions}
Our contributions are threefold. 

First, in Section~\ref{sec:formulation}, we formulate two types of distributed multi-depot vehicle routing problems, offline and online, with the critical and novel feature that no communications are allowed among the servers, achieved by a static partition as mentioned above and explained in detail in Definition~\ref{definition:par}.
The class of problems considered here are of both theoretical and practical value. 
Theoretically, it is interesting to understand the role played by communications among servers (or the absence thereof) by quantifying how well a static partition can perform compared to the optimal solution that is dynamic and request set dependent (and hence requires full communication capabilities among the servers).
Practically, communications among the servers can be costly (overhead, time lag, inefficiency etc.), and hence a close-to-optimal solution without communications can be a highly attractive alternative in practical deployment. (Our results in this paper suggest, as an initial step, that searching for such alternatives can be worthwhile, discussed more later.)  
As shown in Theorem~\ref{thm:relations-distributed-algorithms}, we can restrict our attention to distributed offline problems and to deriving bounds for approximation ratios therein, since the approximation ratio for the offline problem and the competitive ratio for the online problem differ by at most an additive factor of $2$.

Second, for the distributed offline problem, we present several partition schemes that require only polynomial time and space and characterize their approximation ratios.
In Section~\ref{sec:general}, 
we show that the approximation ratio for the Voronoi partition is $m$, where $m$ is the number of servers throughout this paper. 
In Section~\ref{sec:line}, we consider a class of depots that form a line in any metric space; in this case, we give a partition scheme with an approximation ratio of $O(\log m)$. 
Finally, in Section~\ref{sec:bounded}, we consider the case where the ratios between the distances of any two pairs of depots are upper bounded by a sublinear function of $m$ $(f(m))$; in this case, we give a partition scheme with an approximation ratio of $\Theta(f(m))$.

Third, adding to the value of the formulations in this paper is a wide array of open problems with rich opportunities for future explorations (discussed in Section~\ref{sec:open}), with the central one being Open Problem~\ref{open-problem-o(m)}. 
If the answer is positive, then it
will have a fundamental impact in how we view the role played by communications in routing problems. 
This is so particularly because the prior predominant paradigm has mostly been to use dynamic request-assignment scheme via restricted (i.e., local) communications intelligently, with the implicit assumption that local communications are always needed: a completely natural and reasonable assumption at the outset. Our initial results on the two deliberately chosen classes of depot configurations (see Remark~\ref{depots}) indicate that theoretical endeavors along this direction (of finding such a surprisingly good static partition) are worthwhile.
We view the proposed open problems as an open invitation towards this goal.

\subsection{Related Work}\label{sec:related_work}
There has been extensive research on problems related to ours.
Here we place our work in the broader such context by giving a review (in no way complete) of past relevant work, categorized into three classes as follows.  

\textbf{Single-server.}
Given a particular assignment of requests to servers, the problem then reduces to several single-server problems, each of which is the classical \emph{Traveling Salesman Problem} (TSP):
the single server needs to find the optimal order to visit the assigned requests and to return to the depot so that the travelled distance is minimized. 
If $P\neq NP$, no polynomial-time algorithms can solve the TSP~\cite{RefWorks:175}, and they cannot approximate the solution with a ratio better than $117/116$~\cite{RefWorks:169,RefWorks:168}.
On the other hand, \citeauthor{RefWorks:174}~\cite{RefWorks:174} give an $1.5$-approximated algorithm. 
For cases where the metric space is the Euclidean plane~\cite{RefWorks:166,RefWorks:165} or induced from a unit-weight graph~\cite{RefWorks:567,RefWorks:568,RefWorks:566,RefWorks:565}, $(1+\epsilon)$-approximated algorithms (for any $\epsilon>0$)  and $1.4$-approximated algorithms exist respectively.
See \cite{RefWorks:551} for a comprehensive review of hardness results and the classical analysis of heuristic solutions to the TSP and related problems.

\citeauthor{RefWorks:183}~\cite{RefWorks:183} consider the online TSP, in which the requests are revealed incrementally and a server that travels with a unit speed limit is to visit all the requests so as to minimize the returning time (to the original depot). 
The performance measure is the \emph{competitive ratio}, the worst case ratio between the proposed algorithm and an optimal offline algorithm that knows the locations and release dates of all requests from the start.
\citeauthor{RefWorks:183}~\cite{RefWorks:183} have proposed a $2$-competitive algorithm and proved that it is the best possible deterministic online algorithm.  


\textbf{Multi-server Single-depot.}
This class of problems also admits the natural division into offline problems and online problems, 
with the defining feature being that all the servers
share the same depot.  
The offline problems consist primarily of two scenarios. 
In the first scenario, each server is required to visit at least one request. This is known as the multiple traveling salesman problem (see~\cite{RefWorks:549} for a review). 
In the second scenario, each server has a limited capacity, and thus multiple servers are required to visit all the requests. This is known as the capacitated vehicle routing problem. See~\cite{RefWorks:548,RefWorks:547,RefWorks:546} for a review and several heuristics for this problem.

For the online problems, the usually studied objective is to minimize the returning time of the last server (i.e., a Min-Max formulation). 
\citeauthor{RefWorks:52}~\cite{RefWorks:52} have proposed a centralized deterministic algorithm with a competitive ratio of $2$ (the best possible solution).
\citeauthor{RefWorks:180}~\cite{RefWorks:180} consider a variant of this problem in which the cost of an algorithm is measured by the time when the last request is visited and the servers are not required to return to the depot. 
For this problem, when the online algorithm has $k$ servers and the offline algorithm used to define the competitive ratio has only $k^{\star}$ servers ($k^{\star}\leq k$), they propose a centralized deterministic online algorithm with competitive ratio $1+\sqrt{1+1/2^{\lfloor k/k^{\star}\rfloor -1}}$. 


Another interesting and related online problem is the online $k$-server problem, first proposed by \citeauthor{RefWorks:366}~\cite{RefWorks:366}. Here, the ambient metric space is typically a network that consists of a finite number of points (say $n$): the servers can be thought of as moving on a discrete graph.  The requests are revealed incrementally. When a request is revealed, one of the $k$ servers must move to the location of the request instantly, before knowing the subsequent requests. In this model, each request is associated with a release order (as opposed to a release date). The objective is to minimize the total travelled distance of all servers. 
The performance measure used in the literature is the competitive ratio where additive factors independent of the problem instance are allowed (in this paper, such additive factors are not allowed).
%
%

\citeauthor{RefWorks:366}~\cite{RefWorks:366} then conjecture that the lowest competitive ratio for a general ($n$-point) metric space is $k$ (the $k$-server conjecture). As an initial step, they prove that the lower bound holds, and the upper bound holds for special cases where $k=2$ or $n=k+1$. 
As substantial progress along this line, \citeauthor{RefWorks:365}~\cite{RefWorks:365} have proposed the Work Function Algorithm (WFA) and showed that it has a competitive ratio of at most $2k-1$. It is still an open question whether WFA is $k$-competitive. 
For the case where randomization on the actions of online algorithms is allowed, the randomized $k$-server conjecture states that the competitive ratio of the best randomized online algorithm is $\Theta(\log k)$. 
The randomized $k$-server conjecture holds for paging, a special case of the $k$-server problem in which the distance between any pair of points is one. 
For paging, \citeauthor{RefWorks:370}~\cite{RefWorks:370} show that the competitive ratio of the best randomized algorithm is at least the $k^{\text{th}}$ harmonic number $H_k$ ($H_k=1+\frac{1}{2}+\dots +\frac{1}{k}$). On the other hand, \citeauthor{RefWorks:368}~\cite{RefWorks:368} construct a randomized algorithm that achieves a competitive ratio of $H_k$. In addition, using a primal-dual approach, $O(\log k)-$competitive randomized online algorithms can be derived for weighted paging~\cite{RefWorks:590, RefWorks:591}.
While the lower bound of paging can be applied to the $k$-server problem, the competitive ratio of the best randomized online algorithm for the online $k$-server problem is still $2k-1$, the same as the deterministic case. For the case where the competitive ratio is allowed to be dependent on the size of the metric space (say $n$), \citeauthor{RefWorks:589}~\cite{RefWorks:589} develop a randomized online algorithm with a polylogarithmic competitive ratio $O(\log^3 n \log^2 k \log \log n)$.

The aforementioned results (for multi-server single-depot) require centralized algorithms with the knowledge of all the released requests, and hence full communication capabilities among the servers. \citeauthor{RefWorks:367}~\cite{RefWorks:367} consider the distributed version of the online $k$-server problem, in which communications are allowed but induce some costs. \citeauthor{RefWorks:367} give a translation between a centralized algorithm and a distributed one. 

\textbf{Multi-server Multi-depot.}
This class of problems are most relevant to our current formulations. 
Due to the computational complexity, few algorithms that find the optimal solution have been given~\cite{RefWorks:571,RefWorks:579,RefWorks:539} for this case. 
Researchers have focused more on finding heuristic algorithms that find sub-optimal solutions quickly. 

There are primarily three strategies that are commonly used in existing heuristic algorithms.
 The first commonly-used strategy is to first assign each request to the server at the nearest depot, and then refine the solution in a centralized fashion~\cite{RefWorks:572,RefWorks:538,RefWorks:573,RefWorks:581}. 
The second commonly-used strategy is to first assign each request to a server in a centralized fashion, and then determine the route of each server without modifying the assignment of the requests~\cite{RefWorks:541}. 
The third commonly-used strategy is to start from calculating a sub-optimal TSP tour that visits all requests. 
The TSP tour is then divided into pieces, and each piece of the tour is assigned to a server. After that, refinement of the solution is applied~\cite{RefWorks:353,RefWorks:570}. 
There are no performance guarantees for the heuristic algorithms mentioned above. 
In addition, either a centralized planner or communications between servers are required for these heuristic algorithms.

Probabilistic formulations of the multi-depot vehicle routing problems have also been considered. \citeauthor{RefWorks:342}~\cite{RefWorks:342} give a polynomial-time algorithm in which each request is assigned to the nearest depot, and then apply probabilistic analysis to the algorithm. 
They assume that the locations of the requests are independently and uniformly distributed in the unit square and all requests have the same demand, and show that the proposed algorithm achieves an approximation ratio strictly less than $2$ almost surely as the number of requests goes to infinity. 
The focus of our work here is rather different: we do not restrict our attention to any particular metric space, and we do not impose any assumptions on the number of requests nor impose any probabilistic assumptions on the locations of the requests.


The core idea of static partition, a central quantity we study in this paper, lies in dividing the ambient metric space into different regions and letting each server be responsible for requests in the corresponding region.
This idea has been applied to balance the load (defined to be the length of the tour in the solution to the TSP) of servers when the ambient metric space is a compact set in the two-dimensional Euclidean space.
Given the probability distribution of the locations of the requests, \citeauthor{RefWorks:545}~\cite{RefWorks:545} find a partition of the metric space such that the load is almost surely the same for all servers as the number of requests approaches infinity. 
When the precise probability distribution is unknown but some first and second order statistics are given, \citeauthor{RefWorks:516}~\cite{RefWorks:516} find a way to partition the metric space such that the load is most balanced under the worst-case distribution.
We note that there are several key differences from our formulations.
In addition to having different objectives (Min-Max v.s. Min-Sum), the work mentioned above considers only the asymptotic case where the number of requests approaches infinity. Moreover, the partition of the metric space is based on the distribution of the locations of the requests rather than the locations of the depots.



\section{Problem Formulation}\label{sec:formulation}
In this section, we formulate two distributed multi-depot routing problems, offline and online, where communications between servers are not allowed. In Section~\ref{sec:offline_problem}, we formulate the distributed offline problem. In Section~\ref{sec:online_problem}, we formulate the distributed online problem and state the relation between the competitive ratio of the distributed online problem and the approximation ratio of the distributed offline problem. 
For simplicity, our formulations belong to the uncapacitated vehicle routing setting. However, all of our results and analysis are applicable for the capacitated vehicle routing setting when replenishment is allowed at every depot for each server. Due to space limitation, we omit the details.

\subsection{The Distributed Offline Problem}\label{sec:offline_problem}
In the distributed offline problem, there are $m$ servers in the ambient metric space $(\mathbb{M},d)$; each of them has its initial location in one of the $m$ corresponding distinct depots $x_1, x_2, \dots , x_m \in \mathbb{M}$. 
A problem instance consists of a finite list $I$ of requests in the metric space: $I = \{ l_i \in \mathbb{M}\}_{i=1}^n$, for some positive $n$ that can vary in different problem instances. 
Each request needs to be visited by one server. The $m$ servers, which must start and end at their corresponding depots, aim to visit all requests in the most efficient way: here measured by the sum of all distances travelled by all the servers. 
Since the entire problem instance is known when finding a solution (as opposed to requests coming to be known incrementally in an online fashion as discussed in Section~\ref{sec:online_problem}), hence the name ``offline".

On the team level, the central question that immediately arises is finding a good partition of the request set, i.e., which server covers which requests; here and onwards, a partition of the request set is understood to be a disjoint collection of $m$ sets $S_1,\dots ,S_m$ whose union is $I$. 
Note that if server $i$ is (somehow) assigned to a particular subset $S_i$ of requests, then $i$'s optimal action is to, at least in principle, compute the solution $TSP_i (S_i)$ to the traveling salesman problem (the shortest tour that visits each of the requests in $S_i$ subject to the initial and final depot location $x_i$). 

The question then, at first sight, becomes how to find an efficient (or perhaps, in some sense, optimal) partition. 
However, in characterizing an efficient partition, we necessarily need to make assumptions on allowable partition-selection schemes, such as how consensus (on the assignment) is reached among the servers and what communication or collaboration process is allowed. 
If we allow for full collaboration/communication among the servers, then it effectively becomes a centralized planning problem where a central computational unit decides on the optimal partition
$\{S_i^{OPT} \}_{i=1}^m$ that achieves the minimum total cost $OPT(I)$, where
\begin{equation}\label{opt}
OPT(I) \triangleq \min_{S_1,\dots, S_m| \bigcup_{i=1}^m S_i = I} TSP_i(S_i),
\end{equation}
and then distributes the partition to all the servers in some way.

Evidently, the cost achieved in~(\ref{opt}) is the best one can hope for. 
However, there are at least two drawbacks with this formulation.
First, it is instantly clear that this problem is computationally intractable, even disregarding the NP-hardness of the TSP.
The second drawback lies in the strong assumption on the communications involved on the servers' part; moreover, for the case where the requests are revealed incrementally as defined in Section~\ref{sec:online_problem}, such communications need to happen every time a new request is revealed, since the partition of the requests is intrinsically dependent on the entire problem instance. 

Motivated by these two concerns, we naturally wonder if it is possible to find a good static partition of the entire metric space (that depends only on the locations of the depots), which then induces a partition for any request set.
Under this setting, each server needs only visit the requests that fall into its assigned region (and hence no communications between servers are required).  
The following definition formalizes this distributed partition scheme that is static in nature.  

\begin{definition}\label{definition:par}
A distributed partition scheme \textbf{par} is a function that, given the $m$ depot locations $x_1, \dots, x_m$, assigns each server $i$
to a region $M^{\textbf{par}}_i$. $$\textbf{par}(x_1, x_2, \dots, x_m) = (M^{\textbf{par}}_1, M^{\textbf{par}}_2, \dots, M^{\textbf{par}}_m),$$ where $M^{\textbf{par}}_i \subset \mathbb{M}, M^{\textbf{par}}_i \cap M^{\textbf{par}}_j = \emptyset, \forall i, j$, and $\bigcup_{i=1}^{m} M^{\textbf{par}}_i = \mathbb{M}$.
\end{definition}

\begin{remark}
We note that a static partition is sufficient for achieving the no-communication requirement, but not necessary.
In particular, we can consider a time-dependent partition (also independent of the problem instance) as follows.
$$\textbf{par}(x_1, x_2, \dots, x_m, t) = (M^{\textbf{par}}_1(t), M^{\textbf{par}}_2(t), \dots, M^{\textbf{par}}_m(t)).$$
In this setting, request $(r_j,l_j)$ is assigned to server $i$ if and only if $l_j \in M^{\textbf{par}}_i(r_j)$. Since the partition does not depend on the request set, no communications are required. \label{time_dependent_partition}
\end{remark}

Each partition scheme \textbf{par} then induces a distributed algorithm that has the cost function $DIS^\textbf{par}	$, given by
$$ DIS^\textbf{par}(I) \triangleq \sum_{i=1}^m TSP_i(S^\textbf{par}_i)$$
 where $S^\textbf{par}_i = M^\textbf{par}_i\cap I$ is the set of requests assigned to server $i$ under the partition scheme \textbf{par}.
For notational convenience, we often drop the dependence on the specific partition scheme used in the   regions $M_i$, the sets of requests $S_i$, and the cost function $DIS$ when the context shall make it clear which partition scheme we are using.
We use $DIS^\textbf{par}(I)$ as a means to measure the performance of the partition scheme \textbf{par} on the problem instance $I$. Therefore, we study the value but are not concerned about how the value can be computed. For computing the value, the reader is referred to the related work regarding the TSP, as discussed in Section~\ref{sec:related_work}.
	
We note here that the centralized-planning formulation is not entirely useless. The optimal cost $OPT(I)$ defined in \eqref{opt} is a good (and ambitious) comparison metric against which we can evaluate the solution quality of a specific partition scheme $\textbf{par}$. We wish the resulting partition from the distributed problem to be ``close" to $OPT(I)$, and the closer the better, as formalized below.

\begin{definition}
A distributed partition scheme \textbf{par} is $\alpha(m)$-approximated, if for all instances $I$ of any size $n$, $DIS(I)\leq \alpha(m) OPT(I)$. The smaller $\alpha(m)$, the better the partition scheme \textbf{par}.
\end{definition}

We emphasize here that the central problem in the distributed offline problem lies in finding a good partition scheme. 
Theorem~\ref{thm:voronoi} gives an $m$-approximated partition, hence establishing an initial benchmark. However, it is far from obvious whether any sublinear approximation ratio can be achieved in the general case. We take such a partition scheme to be an ambitious goal.

\subsection{The Distributed Online Problem}\label{sec:online_problem}
The crucial (and standard) feature in the distributed online problem is that each request is associated with a release date. 
More precisely, a problem instance $I$ is $$ I = \{ (r_1,l_1), \dots, (r_n,l_n)\mid r_1\leq r_2 \leq \dots \leq r_n , r_j \in \mathbb{R}_{\geq 0}, l_j \in \mathbb{M} \}, $$
where $r_j$ and $l_j$ are the released date and the location of request $j$, respectively, and the number of requests $n$ can vary in different problem instances as in the distributed offline problem. 
All the servers are assumed to have a unit speed limit and are to collectively visit all requests after or at their release dates. 

Due to the additional feature of the requests being online, a static partition of the requests is inadequate in specifying the corresponding cost and how each server moves must also be respected.
To formalize it, we first consider a feasible algorithm $ALG$ that determines the location of each server $i$ at time $t\in \mathbb{R}_{\geq 0}$, denoted $l^{ALG}(i,t)$, subject to the initial location constraint $l^{ALG}(i,0) = x_i$ and the unit speed limit constraint $d(l^{ALG}(i,t_2),l^{ALG}(i,t_1))\leq t_2-t_1$ for all $t_2>t_1\geq 0$. 
Under algorithm $ALG$, the completion time of request $j$, denoted $c^{ALG}_j$, is defined to be the earliest time when one of the servers arrives at the location of the request after or at its release date, i.e.,
$$ c^{ALG}_j \triangleq \inf_{t \geq r_j} \{t \mid \exists i, l^{ALG}(i,t)=l_j\}.$$ 
The cost incurred by server $i$, denoted $ALG_i$, is defined to be the earliest time at which it returns to its depot after all the requests have been served, i.e., 
\begin{align} ALG_i \triangleq \inf_{t\geq\max_{j=1}^n c^{ALG}_j} \{ t|l^{ALG}(i,t)=x_i \}.\label{def:ALG_i}\end{align}

The objective is to minimize the sum of all the time costs, given by
$$ ALG(I) \triangleq \sum_{i=1}^m ALG_i.$$

To meet the no-communication requirement, we now specialize the $ALG$ described above to the current distributed online setting. 
We specify a distributed online algorithm $DOA$ by prescribing the manner in which the trajectories $l^{DOA}(i,t)$ are determined.
 A distributed online algorithm $DOA$ needs to accomplish two tasks. 
First, as in the offline case, $DOA$ must specify a partition scheme $\textbf{par}$ as described in Definition~\ref{definition:par}. 
All the requests that appear in a given region will then only be known by the server to which that region is assigned (and hence no communications between servers are involved). 
Second, $DOA$ needs to decide, for each server $i$, the trajectory $l^{DOA}(i,t)$ based on the 
partially revealed problem instance up to time $t$ that is in $M_i$, i.e.,
$$ I^i_t \triangleq \{ (r_j,l_j)| (r_j,l_j)\in I, r_j \leq t, l_i \in M_i \}.$$

Again, we quantify the performance of a distributed offline algorithm $DOA $ by comparing the cost of a distributed online algorithm to the optimal offline cost $OPT(I)$ that is obtained with the knowledge of the entire problem instance $I$ from the start and full communication capabilities, as formalized below. 
Note that $OPT(I)$ in the online case is different from the one in the offline case because the release-date constraints need to be satisfied for the online problem.

\begin{definition}
A distributed online algorithm $DOA$ is $c(m)$-competitive, if for all instances $I$ of any size $n$, $DOA(I)\leq c(m) OPT(I)$. The smaller $c(m)$, the better the online algorithm $DOA$.
\end{definition}
\begin{remark}
In this paper, we disallow the approximation ratios and the competitive ratios to depend on $n$ by insisting that $DIS(I)\leq \alpha(m) OPT(I)$ and $DOA(I)\leq c(m) OPT(I)$ hold for all $n$.
We do so because we focus on the worst-case problem instance $I$ that may have $n$ significantly bigger than $m$.
More generally, one can allow the ratios to depend on both $m$ and $n$. Depending on the particular setting (i.e., the relation between $m$ and $n$), a good approximation ratio (and competitive ratio) can be decided accordingly.
\end{remark}

The following theorem (proved in Appendix~\ref{proof:thm:relations-distributed-algorithms})
indicates that solving the distributed offline problem is effectively equivalent to solving the distributed online problem. We will therefore restrict our attention to the offline problem in this paper.

\begin{theorem} \label{thm:relations-distributed-algorithms}
If there is an $\alpha-$approximated partition scheme \textbf{par} for the distributed offline problem, then there exists an $(\alpha+2)-$competitive distributed online algorithm for the distributed online problem.
\end{theorem}

\section{Static Partition Schemes for the Distributed Offline Problem}\label{sec:special_configurations}
\subsection{General Depot Configuration}\label{sec:general}

We consider the Voronoi partition $VOR$ in which a point $p\in \mathbb{M}$ is in $M^{VOR}_i$ if $x_i$ is the nearest depot for $p$ (ties broken arbitrarily). The following theorem (proved in Appendix~\ref{proof:thm:voronoi}) indicates that $O(m)$ partition schemes exist for general depot configurations.

\begin{theorem}\label{thm:voronoi}
The approximation ratio of the Voronoi partition $VOR$ is exactly $m$.
\end{theorem}

The central open problem (Open Problem~\ref{open-problem-o(m)}) is whether there exists a partition scheme that does better (sublinear in $m$) than the baseline Voronoi partition in the general case.
To make progress towards this direction, we consider two special classes of depot configurations and identify such sublinear partition schemes. We note that the Voronoi partition, when specialized to these two classes of depot configurations, still gives an approximation ratio of $\Omega(m)$ (see Appendix~\ref{examples} for examples). 
\begin{remark}
These two classes of depot configurations are deliberately chosen as they stand on the two extremes of the general case: the line case (Section~\ref{sec:line}) is the most ``stretched-out'' depot configuration and the bounded-ratio case (Section~\ref{sec:bounded}) is the most ``clustered'' depot configuration.
The existence of sublinear partition schemes, although different, in the two extremes leads us to conjecture that there exists a sublinear partition scheme for the general case.\label{depots}\end{remark}
\subsection{Depots on A Line}\label{sec:line}
Here we consider the case in which the depots form a line in the metric space, i.e., $d(x_i, x_j)+d(x_j, x_k)=d(x_i,x_k)$ for any $1\leq i < j < k \leq m$. For this configuration of depots, we give a partition scheme (the Level partition $LEV$) with an approximation ratio of $O(\log m)$.
It is not clear whether our analysis is tight: the proposed partition scheme may have a lower asymptotic approximation ratio. 
 
\textbf{The Level Partition.}
For notational convenience, we assume that there are $2^k+1$ servers for some integer $k \ge 0$, indexed in order as $0,1,2,\dots, m-1 (= 2^k)$.  
If the number of servers is not $2^k+1$, then we can duplicate the last depot for the filling: creating enough copies (at most $m-3$) of depot $m-1$ so that the total number of depots is brought to $2^k+1$. The partition scheme and analysis still apply.
 
For each integer $l=0,1,\dots, k-1$, let $\mathbb{N}_l$ denote the set of integers between $1$ and $2^k-1$ that are multiples of $2^{l}$ but not $2^{l+1}$, i.e., $\mathbb{N}_l \triangleq \{ 2^{l}(2t+1)| t=0,1,2,\dots 2^{k-l-1}-1\}$. As special cases, we denote $\mathbb{N}_k \triangleq \{ 2^k\}$ and $\mathbb{N}_{k+1} \triangleq \{ 0 \}$.
We say that server $i$ is in level $l$ if $i \in \mathbb{N}_l$,
hence the name Level partition.

For $i=1,2,\dots,  2^k-1$, we denote by $\tau_i$ the intersection of two particular closed circular disks:\begin{align*}
\tau_i  \triangleq & \{p\in \mathbb{M} \mid d(p,x_{i-2^l}) \leq d(x_{i-2^l}, x_{i}) + \lambda d(x_{i+2^l},x_i)  \} 
\\  \cap & \{p\in \mathbb{M} \mid d(p,x_{i+2^l}) \leq d(x_{i+2^l}, x_{i}) + \lambda d(x_{i-2^l},x_i)  \},
\end{align*}
where $\lambda \triangleq 3/4$ for simplicity (in fact, any fixed constant in $(1/2,1)$ will do) and $l$ is the integer such that $i\in \mathbb{N}_l$. 
For special cases where $i=2^k, 0$, we define $\tau_{2^k}  \triangleq \{ p\in \mathbb{M} \mid d(p,x_{2^k}) \leq \lambda d(x_{0},x_{2^k})  \}$ and $\tau_{0}  \triangleq  \mathbb{M}$.


For each $i=0, 1,2,\dots , 2^k$ with server $i$ in level $l$ (i.e.,\ $i \in \mathbb{N}_l$), we define $M^{LEV}_i$ to be the points in $\tau_i$ but not in $\tau_{i'}$ for any lower-level $i'$, i.e., 
$$ M^{LEV}_i \triangleq \tau_i \setminus \bigcup_{l'=0}^{l-1} \bigcup_{i' \in \mathbb{N}_{l'}} \tau_{i'}.$$
See Figure~\ref{figure:partition} for an illustration of the Level partition with $m=9$.
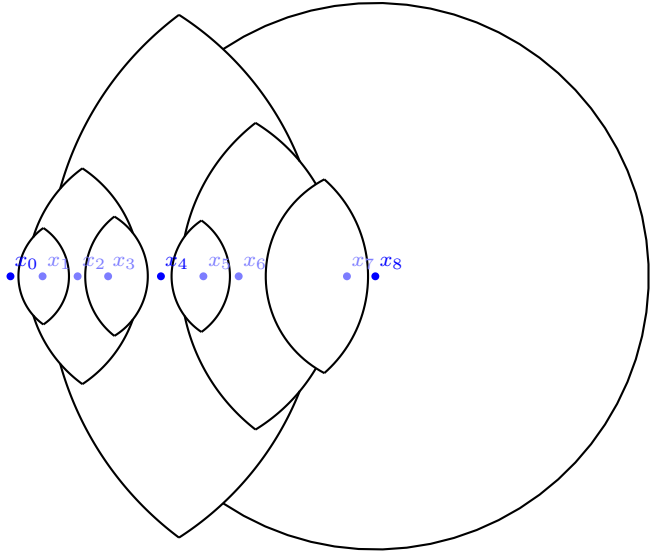
\begin{figure}
	\centering  
\newrgbcolor{xdxdff}{0.49 0.49 1}
\psset{xunit=0.4cm,yunit=0.4cm,algebraic=true,dimen=middle,dotstyle=o,dotsize=3pt 0,linewidth=0.8pt,arrowsize=3pt 2,arrowinset=0.25}

\begin{pspicture*}(-0.75,-9.26)(21.66,9.31)
\parametricplot{-0.9737673392981776}{0.9737673392981779}{1*1.94*cos(t)+0*1.94*sin(t)+0|0*1.94*cos(t)+1*1.94*sin(t)+0}
\parametricplot{2.1901715727438518}{4.0930137344357345}{1*1.97*cos(t)+0*1.97*sin(t)+2.23|0*1.97*cos(t)+1*1.97*sin(t)+0}
\parametricplot{-1.0200656333154132}{1.0200656333154132}{1*2.33*cos(t)+0*2.33*sin(t)+2.23|0*2.33*cos(t)+1*2.33*sin(t)+0}
\parametricplot{2.2343426801820088}{4.048842626997578}{1*2.52*cos(t)+0*2.52*sin(t)+5|0*2.52*cos(t)+1*2.52*sin(t)+0}
\parametricplot{2.200722801352747}{2.8449435983541216}{1*4.44*cos(t)+0*4.44*sin(t)+5|0*4.44*cos(t)+1*4.44*sin(t)+0}
\parametricplot{3.4382417088254646}{4.08246250582684}{1*4.44*cos(t)+0*4.44*sin(t)+5|0*4.44*cos(t)+1*4.44*sin(t)+0}
\parametricplot{0.34027338489987474}{0.9845639557449344}{1*4.31*cos(t)+0*4.31*sin(t)+0|0*4.31*cos(t)+1*4.31*sin(t)+0}
\parametricplot{5.298621351434652}{5.942911922279712}{1*4.31*cos(t)+0*4.31*sin(t)+0|0*4.31*cos(t)+1*4.31*sin(t)+0}
\parametricplot{-0.944424116253006}{0.9444241162530064}{1*2.29*cos(t)+0*2.29*sin(t)+5|0*2.29*cos(t)+1*2.29*sin(t)+0}
\parametricplot{2.1606835658256918}{4.122501741353895}{1*2.23*cos(t)+0*2.23*sin(t)+7.58|0*2.23*cos(t)+1*2.23*sin(t)+0}
\parametricplot{-0.8466297958941489}{0.8466297958941484}{1*4.3*cos(t)+0*4.3*sin(t)+7.58|0*4.3*cos(t)+1*4.3*sin(t)+0}
\parametricplot{2.052956074800722}{4.230229232378864}{1*3.64*cos(t)+0*3.64*sin(t)+12.12|0*3.64*cos(t)+1*3.64*sin(t)+0}
\parametricplot{2.2343852162131066}{2.937048268568464}{1*6.47*cos(t)+0*6.47*sin(t)+12.12|0*6.47*cos(t)+1*6.47*sin(t)+0}
\parametricplot{3.3461370386111224}{4.04880009096648}{1*6.47*cos(t)+0*6.47*sin(t)+12.12|0*6.47*cos(t)+1*6.47*sin(t)+0}
\parametricplot{0.5359451922462065}{1.0201116440717135}{1*5.99*cos(t)+0*5.99*sin(t)+5|0*5.99*cos(t)+1*5.99*sin(t)+0}
\parametricplot{5.263073663107873}{5.74724011493338}{1*5.99*cos(t)+0*5.99*sin(t)+5|0*5.99*cos(t)+1*5.99*sin(t)+0}
\parametricplot{5.284287180547575}{5.90493735664418}{1*10.34*cos(t)+0*10.34*sin(t)+0|0*10.34*cos(t)+1*10.34*sin(t)+0}
\parametricplot{0.37824795053540533}{0.9988981266320114}{1*10.34*cos(t)+0*10.34*sin(t)+0|0*10.34*cos(t)+1*10.34*sin(t)+0}
\parametricplot{2.2144921677671663}{2.87020208504102}{1*10.87*cos(t)+0*10.87*sin(t)+12.12|0*10.87*cos(t)+1*10.87*sin(t)+0}
\parametricplot{3.4129832221385663}{4.06869313941242}{1*10.87*cos(t)+0*10.87*sin(t)+12.12|0*10.87*cos(t)+1*10.87*sin(t)+0}
\parametricplot{-2.160885911767373}{2.160885911767373}{1*9.09*cos(t)+0*9.09*sin(t)+12.12|0*9.09*cos(t)+1*9.09*sin(t)+0}
\begin{scriptsize}
\psdots[dotstyle=*,linecolor=blue](5,0)
\rput[bl](5.14,0.24){\blue{$x_4$}}
\psdots[dotstyle=*,linecolor=blue](0,0)
\rput[bl](0.14,0.24){\blue{$x_0$}}
\psdots[dotstyle=*,linecolor=xdxdff](2.23,0)
\rput[bl](2.39,0.24){\xdxdff{$x_2$}}
\psdots[dotstyle=*,linecolor=xdxdff](1.06,0)
\rput[bl](1.23,0.24){\xdxdff{$x_1$}}
\psdots[dotstyle=*,linecolor=xdxdff](3.24,0)
\rput[bl](3.4,0.24){\xdxdff{$x_3$}}
\psdots[dotstyle=*,linecolor=blue](12.12,0)
\rput[bl](12.27,0.24){\blue{$x_8$}}
\psdots[dotstyle=*,linecolor=xdxdff](6.41,0)
\rput[bl](6.58,0.24){\xdxdff{$x_5$}}
\psdots[dotstyle=*,linecolor=xdxdff](7.58,0)
\rput[bl](7.74,0.24){\xdxdff{$x_6$}}
\psdots[dotstyle=*,linecolor=xdxdff](11.18,0)
\rput[bl](11.34,0.24){\xdxdff{$x_7$}}
\end{scriptsize}
\end{pspicture*}
  \vspace{-1em}
\caption{The Level Partition}   \vspace{-2em}
\label{figure:partition}
\end{figure}

Our main result regarding the Level partition is the following theorem.
\begin{theorem} \label{thm:competitive-level}
	$DIS^{LEV}(I) \leq O(\log m) OPT(I).$ 
\end{theorem}
In order to prove Theorem~\ref{thm:competitive-level}, we divide the metric space $\mathbb{M}$ into $k+2$ different levels of regions $L_0,L_1, \dots L_{k+1}$, where $L_l$ is defined to be the union of regions assigned to servers in level $l$, i.e., $$L_l \triangleq \bigcup_{i \in \mathbb{N}_l}M^{LEV}_i.$$

The definition of $\{ L_l \}_{l=0}^{k+1}$ is used to prove the following lemma. 
\begin{lemma}\label{lemma:each_level}
	There exists a constant $\rho$ such that $$LEV (I )\leq \rho OPT(I )$$ if all requests are in the same level $L_l$, i.e., $I \subseteq L_l$, for any integer $l=0,1,\dots, k+1$. 
\end{lemma}
Lemma~\ref{lemma:each_level} is sufficient for proving Theorem~\ref{thm:competitive-level}.
\begin{proof}[Proof of Theorem~\ref{thm:competitive-level}]
$$LEV(I) =\sum_{l=0}^{k+1}LEV(I\cap L_l)\leq \rho \sum_{l=0}^{k+1}OPT(I \cap L_l)\leq  \rho (k+2) OPT(I)= O(\log m) OPT(I)$$ where the last inequality is due to $ I \cap L_l \subseteq I$ for any $l=0,1,\dots, k+1$.
\qedhere
\end{proof}
In the rest of this section, we are going to describe the proof sketch for Lemma~\ref{lemma:each_level} (see Appendix~\ref{proof:lemma:each_level} for the complete proof). 
To prove Lemma~\ref{lemma:each_level}, we first find the relation between the cost of the optimal centralized algorithm that uses all servers and that of a centralized algorithm that uses only servers with indexes in $\mathbb{N}_l$. We call such algorithms \emph{responsible} algorithms because servers in level $l$ are responsible for all requests located in $L_l$. In particular, we have the following lemma (proved in Appendix~\ref{proof:lemma:opt-to-type1}).
\begin{lemma}\label{lemma:opt-to-type1} For any integer $l$, when $I\subseteq L_l$, there exists a responsible algorithm $RES$ whose cost is at most $1+1/g$ times of the cost of the optimal solution $OPT(I)$ for any integer $l=0,1,\dots, k+1$ where $g$ is defined to be $1/30$ for simplicity.
\end{lemma}

Once we have Lemma~\ref{lemma:opt-to-type1}, what is remaining is to find a relation between the cost of the given responsible algorithm $(RES(I))$, and the cost induced by the Level partition $(DIS^{LEV}(I))$. 
For each server $i\in \mathbb{N}_l$, we consider the sequence of the regions $\tau_j$ with $j\in \mathbb{N}_l$ visited by server $i$.
Note that for server $i$, the sequence begins and ends at $\tau_i$, and if the sequence only consists of one region $( \tau_i )$ for all $i$, then the responsible algorithm is the same as the optimal distributed algorithm induced by the Level partition. 
Therefore, we encounter an issue only when the sequence contains multiple regions for some server $i$.
We say that a responsible algorithm is \emph{non-oscillating} if a zigzag with at least four regions, i.e, pattern $(\tau_j, \tau_{j'}, \tau_j , \tau_{j'} )$, does not occur in the sequence of the route of any server $i$. 
In a non-oscillating responsible algorithm, it is allowed for each server $i$ not to travel through the optimal TSP tour given the set of requests that are assigned to it.
Given a responsible algorithm $RES$, we create a non-oscillating responsible algorithm $NOS$ as an intermediate step for comparing $RES(I)$ with $DIS^{LEV}(I)$, as shown in the following two lemmas (proved in Appendix~\ref{proof:lemma:type1-to-non-oscillating} and~\ref{proof:lemma:non-oscillating-to-alg} respectively).
\begin{lemma}\label{lemma:type1-to-non-oscillating}
Given a responsible algorithm $RES$, there exists a non-oscillating responsible algorithm $NOS$ whose cost is at most four times of the cost of $RES$ when $I\subseteq L_l$ for any integer $l=0,1,\dots, k+1$. \end{lemma}
\begin{lemma}\label{lemma:non-oscillating-to-alg}
The cost $DIS^{LEV}(I)$ of the distributed algorithm based on the Level partition is at most $\frac{25-5\lambda -6 \lambda^2}{1-\lambda}$ times of the cost of any given non-oscillating responsible algorithm $NOS$ when $I\subseteq L_l$ for any integer $l=0,1,\dots, k+1$. \end{lemma}
\subsection{Bounded Ratios between Distances of Depots}\label{sec:bounded}
Here we consider the case where the ratios between the distances of depots are bounded above by a sublinear function $f$ of $m$, i.e., $$\frac{\max_{i,j}\{d(x_i,x_j)\}}{\min_{i,j}\{d(x_i,x_j)\}} \leq f(m).$$ 
We give a partition scheme (Local partition $LOC$) with an approximation ratio of $\Theta(f(m))$ (and hence for it to be useful, we assume $f(m)=o(m)$).

\textbf{The Local Partition.}
Fix an arbitrary ordering of the servers, $i=1,2,\dots , m$.
The region $M^{LOC}_i$ that each server $i$ is responsible for centers around, except for server $m$, a local region around its depot:
$$M^{LOC}_i \triangleq \left\{p\in \mathbb{M} \mid d(p,x_i) < \frac{\min_{j',j}\{d(x_{j'},x_j)\}}{4} \right\}, 1 \le i \le m-1,$$
$$ M^{LOC}_m \triangleq \mathbb{M}\setminus \bigcup_{i=1}^{m-1}M^{LOC}_i.$$
\begin{theorem}\label{thm:local_partition}
	The Local partition has an approximation ratio of $\Theta(f(m))$. 
\end{theorem}
See Appendix~\ref{proof:thm:local_partition} for the proof.


\section{Open Problems}\label{sec:open}
Our formulations open up several challenging problems, which we discuss here. We believe that seeking the answers to any of the following open problems, either positive or negative, would be of great theoretical interest and practical use. 

\begin{open problem}\label{open-problem:line}
Does there exist a partition scheme with an approximation ratio of $o(\log (m))$ when the depots form a line?
\end{open problem}
This open problem is immediate: either by a different partition scheme or tighter analysis of our $O(\log (m))$-approximated partition scheme for the line case.

\begin{open problem}\label{open-problem:2d}
Does there exist a partition scheme with an approximation ratio of $o(m)$ when the metric space is the two-dimensional Euclidean space $\mathbb{R}^2$? What about $\mathbb{R}^3$?
\end{open problem}
Since difficulties can potentially arise when one goes from $1D$ to $2D$ and from $2D$ to $3D$ (where going from $3D$ to higher dimensions is typically straightforward), answering this question can be very valuable.

\begin{open problem}\label{open-problem-configurations-O(1)}
For a fixed sublinear function $f$, under what configurations of depots, does there exist a partition scheme with an approximation ratio of $\Theta(f(m))$?
\end{open problem}
This problem is a natural generalization of our bounded-ratio result. In particular, the bounded-ratio configuration works for any sublinear $f$.

\begin{open problem}\label{open-problem-o(m)}
Does there exist a partition scheme with an approximation ratio of  $o(m)$ for any configuration of depots in any metric space?
\end{open problem}
Our conjecture here, as mentioned in Remark~\ref{depots}, is that such sublinear partition schemes exist.

\begin{open problem}\label{open-problem-O(1)}
Does there exist a partition scheme with an approximation ratio of $\Theta(1)$ for any configuration of depots in any metric space?
\end{open problem}

If the answer to Open Problem~\ref{open-problem-O(1)} is yes, then the resulting partition scheme will be a surprisingly good one that completely trivializes the use of communications.
\\\textbf{Variants of The Distributed Online Problem.} 
Due to the offline-online equivalence result (Theorem~\ref{thm:relations-distributed-algorithms}) we have focused exclusively on offline problems. However, there exist variants of the objective function for the online problem where such a reduction is not easily obtained. Below is an example.

Given an online algorithm $ALG$, consider the cost incurred by server $i$ as follows (compare it with Equation~\eqref{def:ALG_i}). $$ ALG_i \triangleq \inf_{t \geq \max_{(r_j,l_j) \in S_i}\{ c^{ALG}_j\} } \{t \mid l^{ALG}(i,t) = x_i \}.$$

For this problem, we are again interested in finding partition schemes for online algorithms that have sublinear competitive ratios. However, it is easy to show that the partition schemes discussed in Definition~\ref{definition:par} cannot provide such online algorithms (consider the problem instance consisting of $m$ requests with $(r_j,l_j)=(TSP(x_1,x_2,\dots, x_m),x_j)$ for all $j$). 
Therefore, alternative partition schemes that do not require communications between servers need to be taken into consideration, which opens interesting research directions. The time-dependent partition schemes (see Remark~\ref{time_dependent_partition}), where the assignment of a request depends on both the location and the release date of the request, are possible candidates for solving this variant of the distributed online problem.

\newpage
\bibliographystyle{mynat} 
\bibliography{bibdata}
\newpage
\appendix
\section{Appendix}
\subsection{Proof of Theorem~\ref{thm:relations-distributed-algorithms}}\label{proof:thm:relations-distributed-algorithms}
\begin{proof}
To simplify the notation, we denote $R^i$ the set of the locations of all requests in $M_i$ and $R^i_t$ the set of the locations of requests in $M_i$ with release dates at most $t$, i.e., 
$$ R^i \triangleq \{l_j:  (r_j,l_j)\in I, r_j \in M_i\}\text{ and }  R^i_t \triangleq \{l_j:  (r_j,l_j)\in I^i_t\}.$$
We consider the following distributed online algorithm $DOA$. 
	\begin{itemize}
		\item The online algorithm $DOA$ adopts the partition scheme $\textbf{par}$.
		\item At any time $t$, if a request is released in $M_i$, then server $i$ stops traveling through the planned route, returns to its depot $x_i$, and then follows the route $TSP_i(R^i_t)$.
	\end{itemize}
We first find an upper bound on the cost of the algorithm $DOA$. Let $j$ denote the index of the request with the maximum release date in $M_i$. 
At the release date $r_j$, the distance between the location of server $i$ and its depot $x_i$ is at most $r_j$, and thus the time when server $i$ returns at its depot and starts the route $TSP_i(R^i_{r_j})$ is at most $2r_j$, which is at most $2r_n$. Moreover, $R^i_{r_j}  = R^i$.
Therefore, $DOA_i \leq 2 r_n + TSP_i(R^i)$. 

Now let us find a lower bounds on the cost of the optimal centralized offline algorithm $OPT(I)$. 
First we notice that the completion time of request $n$ is lower bounded by its release date, i.e., $c^{OPT}_n \geq r_n$, and therefore $OPT_i \geq r_n$ for all $i=1,2,\dots , m$.
As a result, $OPT(I) \geq m r_n$. 
In addition, according to the presumption of this lemma, $\alpha OPT(I) \geq \sum_{i=1}^m TSP_i(R^i)$. 

Combine the results above, we have
\begin{align*}
	DOA (I)= \sum_{i=1}^m DOA_i \leq 2mr_n + \sum_{i=1}^m TSP_i(R^i) \leq 2OPT(I)+\alpha OPT(I).
\end{align*}
Hence, we conclude that the competitive ratio of the distributed online algorithm $DOA$ is at most $ \alpha +2$.
\end{proof}

\subsection{Proof of Theorem~\ref{thm:voronoi}}\label{proof:thm:voronoi}

\begin{proof}
\citeauthor{RefWorks:342}~\cite{RefWorks:342} give an example demonstrating that the Voronoi partition has an approximation ratio of at least $m$.
Therefore, it suffices to prove that the approximation ratio of the Voronoi partition is at most $m$.

We first consider the route of any server (without loss of generality and for convenience, say Server $1$) under the Voronoi partition, and show that the following holds:
\begin{align}\label{claim:server_1}
TSP_1(S^{VOR}_1) \leq OPT(I),
\end{align}
where $S^{VOR}_1$, as defined in Section~\ref{sec:offline_problem}, is the set of requests assigned to Server $1$ under the Voronoi partition. 

Since Server $1$ in Inequality~(\ref{claim:server_1}) can be replaced by any server, we have
	$$ DIS^{VOR}(I) = \sum_{i=1}^m TSP_i(S^{VOR}_i) \leq mOPT(I), $$
which proves the theorem.	
	
We now prove Inequality~(\ref{claim:server_1}) now. Without loss of generality, assume for some integer $k$, $S^{OPT}_i \cap S^{VOR}_1 \neq \emptyset$ for all $i=2,3,\dots, k$, and $S^{OPT}_i \cap S^{VOR}_1 = \emptyset$ for all $i=k+1,k+2,\dots , m$, where $S^{OPT}_i$, as defined in Section~\ref{sec:offline_problem}, is the set of requests assigned to Server $i$ under the optimal partition $OPT$. 

Therefore, for each $i \in \lbrace 2, 3, \dots, k \rbrace$, there exists at least one request $p$ in the set $S^{OPT}_i$ such that $d(x_1 , p) \leq d(x_i , p)$. 
For each $i$, among all such ``closer-to-deopt $x_1$" requests, we denote the first and last requests visited by $TSP_i(S^{OPT}_i)$ (the route of Server $i$ under $OPT$) by $a_i$ and $b_i$ respectively, where $a_i$ and $b_i$ coincide if there is only one such request.

We create as follows a route for Server $1$ that visits all requests in $S^{VOR}_1$ based on the routes of servers $1,2,\dots , k$ under the optimal algorithm $OPT$. 
First of all, let Server $1$ follow $TSP_1(S^{OPT}_1)$, the route of Server $1$ in the optimal algorithm. For each $i = 2,3,\dots , k$, let Server $1$ follow the additional round trip that begins and ends at the depot $x_1$.
		\begin{enumerate}
			\item Travel from $x_1$ to $a_i$.
			\item Travel from $a_i$ to $b_i$ using $TSP_i(S^{OPT}_i)$.
			\item Travel from $b_i$ to $x_1$.
		\end{enumerate}
	It is clear that by doing so, Server $1$ visits all requests in $S^{VOR}_1$. 
	Therefore, the length of the route of Server $1$ defined above is an upper bound of $TSP_1(S^{VOR}_1)$.
	For each $i \in \lbrace 2, 3, \dots, k \rbrace$, the length of the additional route because of requests in $S^{OPT}_i$ is at most $TSP_i(S^{OPT}_i)$ because $ d(x_1,a_i)\leq d (x_i,a_i)$ and $ d(x_1,b_i)\leq d (x_i,b_i)$. 
	Therefore, $$TSP_1(S^{VOR}_1) \leq \sum_{i=1}^k TSP_i (S^{OPT}_i) \leq OPT(I),$$
	which establishes Inequality~(\ref{claim:server_1}).
		\qedhere	
\end{proof}
\subsection{$\Omega(m)$-approximation Examples for The Voronoi Partition}\label{examples}
For the line case, we provide the following example. 
Let the metric space be the two dimensional Euclidean space and $x_i = (0,i)$ for $i=1,2,\dots , m$. 
Let $I$ to be the problem instance that consists of $m$ requests where $l_j = (k,j)$ for some constant $k$. Clearly, $DIS^{VOR}(I)=2km$ and $OPT(I) \leq 2k+2m$. When $k\rightarrow \infty$, the ratio $DIS^{VOR}(I)/ OPT(I)$ approaches $m$.

For the bounded-ratio case, we provide the following example with $f(m)=1$. 
Let the metric space be the $m-$dimensional Euclidean space and $x_i=e_i$ where $e_i$ the $m-$dimensional vector that has the value of $1$ in the $i^{\text{th}}$ dimension and $0$ in all other dimensions. Let $I$ to be the problem instance that consists of $m$ requests where $l_j=\epsilon e_j$ for some constant $\epsilon >0$ for each $j=1,2,\dots , m$. 
Clearly, $DIS^{VOR}(I)=2m(1-\epsilon)$. 
On the other hand, the cost of the optimal algorithm is at most that of the algorithm that assigns all requests to the same server, which is upper bounded by $2+2m\epsilon$. 
Therefore, the ratio $DIS^{VOR}(I)/OPT(I)$ approaches $m$ as $\epsilon \rightarrow 0^+$. 

\subsection{Proof of Lemma~\ref{lemma:each_level}}\label{proof:lemma:each_level}
\begin{proof}
With Lemmas~\ref{lemma:opt-to-type1},~\ref{lemma:type1-to-non-oscillating}, and~\ref{lemma:non-oscillating-to-alg}, we can prove Lemma~\ref{lemma:each_level} as follows.
\begin{align*}
LEV(I) & \leq 		\frac{25-5\lambda -6 \lambda^2}{1-\lambda} NOS(I) \leq 4 \times \frac{25-5\lambda -6 \lambda^2}{1-\lambda} RES(I) 
\\ & \leq \left(1+\frac{1}{g} \right)\frac{100-20\lambda -24 \lambda^2}{1-\lambda} OPT(I) \leq 9000 OPT(I)
\end{align*}	
where $\lambda = 3/4$ and $g=1/30$. Therefore, the lemma holds when $\rho$ is chosen to be, for example, $9000$.
	\qedhere
\end{proof}
\subsection{Proof of Lemma~\ref{lemma:opt-to-type1}}\label{proof:lemma:opt-to-type1}
\begin{proof}
To simplify the discussion, we prove only the cases in which $l=0,1,2,\dots , k-1$. However, it is clear that the proof can be modified and applied to the cases where $l=k, k+1$. 


We define the following responsible algorithm $RES$. 
For each $i'=0,1,2,\dots,2^k$, if $S^{OPT}_{i'}$ is not empty, i.e., there are requests served by server $i'$ in the optimal algorithm $OPT$, then we assign all requests in $S^{OPT}_{i'}$ to the server with the minimum index $i$ in $\mathbb{N}_l$ such that $S^{OPT}_{i'}$ contains a request in $M^{LEV}_i$. 
This is equivalent to the following definition.
$$ S^{RES}_i \triangleq \bigcup_{i' \in A_i}S^{OPT}_{i'}$$ where $$ A_i \triangleq \{ j: S^{OPT}_j \cap M^{LEV}_i \neq \emptyset \text{ and } S^{OPT}_{j} \cap M^{LEV}_{i'} =\emptyset \text{ for all }i' < i, i' \in \mathbb{N}_l \}. $$
Now we define the route of each server $i$ in $\mathbb{N}_l$ separately. 
Let us denote $i_{\max}$ and $i_{\min}$ the maximum and minimum index in $A_i$  respectively. The route of server $i$ travels from $x_i$ to $x_{i_{\max}}$, then to $x_{i_{\min}}$, and finally back to $x_i$ through the shortest path that passes through depots $x_{i'}$ for all $i_{\min}\leq i' \leq i_{\max}$. 
In addition, server $i$ travels through $TSP_{i'}(S^{OPT} _{i'})$ when passing by depot $x_{i'}$ for the first time before going to the next depot if $i'\in A_i$. Clearly, the cost of $RES$ can be calculated as follows.
$$RES(I)  = \sum_{i \in \mathbb{N}_l} \left(\sum_{i' \in A_i} TSP_{i'}(S^{OPT}_{i'}) +2d(x_{i_{\min}},x_{i_{\max}})\right)$$

We will prove the following claim.
\begin{align}
 d(p,x_{i_{\max}}) \geq g d(x_{i},x_{i_{\max}})\label{inequality:distance}
\end{align}
for any $p\in M^{LEV}_{i}$. 

Before proving Claim~(\ref{inequality:distance}), we first show that Claim~(\ref{inequality:distance}) implies the lemma. 

If Claim~(\ref{inequality:distance}) is true, then
\begin{align*}
	TSP_{i_{\max}} (S^{OPT}_{i_{\max}}) \geq 2 d(p,x_{i_{\max}}) \geq 2gd(x_i ,x_{i_{\max}})
\end{align*} for any point $p$ in  $S^{OPT}_{i_{\max}} \cap M^{LEV}_i.$ 
By symmetry, the same result holds for $i_{\min}$.
Therefore, $$ \frac{1}{g}\sum_{i' \in A_i} TSP_{i'}(S^{OPT}_{i'}) \geq \frac{1}{g} (TSP_{i_{\min}} (S^{OPT}_{i_{\min}}) + TSP_{i_{\max}} (S^{OPT}_{i_{\max}})) \geq 2 d(x_{i_{\min}},x_{i_{\max}})$$
when $i_{min}\neq i_{max}$. On the other hand, if $i_{min}=i_{max}$, then $i_{min}=i_{max}=i$. Hence, $$ \frac{1}{g}\sum_{i' \in A_i} TSP_{i'}(S^{OPT}_{i'}) \geq 0 = 2 d(x_{i_{\min}},x_{i_{\max}}).$$

As a result, given Claim~(\ref{inequality:distance}), the lemma can be proven as follows.
\begin{align*}
	RES(I)& = \sum_{i \in \mathbb{N}_l} \left(\sum_{i' \in A_i} TSP_{i'}(S^{OPT}_{i'}) +2d(x_{i_{\min}},x_{i_{\max}})\right)
	\\& \leq \left(1+\frac{1}{g}\right) \sum_{i \in \mathbb{N}_l} \left(\sum_{i' \in A_i} TSP_{i'}(S^{OPT}_{i'}) \right)
	\\& =  \left( 1+\frac{1}{g} \right) OPT(I).\end{align*}

Let us now prove Claim~(\ref{inequality:distance}).
To simplify the notation, let us denote $j=i+2^{l}$. 

First we consider case where $i_{\max}\geq j$. In this case,
\begin{align}
	d(p, x_{i_{\max}}) &\geq d(x_{i-2^{l}},x_{i_{\max}})-d(x_{i-2^{l}},p) \geq d(x_{i-2^{l}},x_{i_{\max}})- (d(x_{i-2^{l}},x_i)+\lambda d(x_i,x_{j})) \nonumber	\\ &=(1-\lambda)d(x_i,x_{j})+d(x_j, x_{i_{\max}}) \label{inequality:not_close_to_x_j1}
\\	 &\geq (1-\lambda)d(x_i,x_{i_{\max}}) \geq g d(x_i,x_{i_{\max}}).  \nonumber
\end{align}

Let us now consider the other case where $i<i_{\max}<j$ (this case is possible only if $l\geq1$). Inequality~(\ref{inequality:not_close_to_x_j1}) with $i_{\max}=j$ gives us
\begin{align}
	d(p, x_j)\geq (1-\lambda)d(x_i,x_{j}),\label{inequality:not_close_to_x_j} 
\end{align} which motivates us to distinguish cases further based on the ratio between $d(x_i,x_{i_{\max}})$ and $d(x_i,x_j)$. 
Define the threshold of the ratio to be $f\triangleq 7/8$ for simplicity. 
In fact, any fixed constant in $(\lambda, 1)$ will do (for possibly a different positive real number $g$).

If $d(x_i,x_{i_{\max}}) > fd(x_i,x_j)$, according to the triangle inequality and Inequality~(\ref{inequality:not_close_to_x_j}), 
	\begin{align*}
	d(p,x_{i_{\max}}) &	\geq d(p,x_j) - d(x_j,x_{i_{\max}}) >
(1-\lambda)d(x_i,x_j)-(1-f)d(x_i,x_j) 
\\& = (f- \lambda) d(x_i,x_j)>(f- \lambda) d(x_i,x_{i_{\max}}) \geq g d(x_i,x_{i_{\max}}).
	\end{align*}
	
What is remaining is the case where $1\leq i_{\max}\leq j-1$ and $d(x_i,x_{i_{\max}}) \leq fd(x_i,x_j)$. This case implies that $$ d(x_{i_{\max}}, x_j) \geq \frac{1-f}{f} d(x_i, x_{i_{\max}}),$$ which we use frequently in the remaining of proof. We distinguish two cases.
\begin{enumerate}
	\item \label{casea}$1\leq i_{\max} \leq i+2^{l-1}$.\\
	In this case, there exists a positive integer $t' \leq l-1$ such that $i+2^{t'-1}\leq i_{\max}\leq i+2^{t'}$. 
According to the definition of $M^{LEV}_i$, $p$ is not in $\tau_{i+2^t}$ for any $t=0,1,\dots , l-1$. Thus, for each $t=t'-1, t', \dots , l-1$, the point $p$ in $M^{LEV}_{i}$ must violate one of the following two constraints.
\begin{align}
 d(p,x_{i}) & \leq d(x_{i}, x_{i+2^t}) + \lambda d(x_{i+2^t},x_{i+2^{t+1}}). \label{constraint:1}
\\ d(p,x_{i+2^{t+1}}) & \leq d(x_{i+2^{t+1}}, x_{i+2^t}) + \lambda d(x_{i},x_{i+2^t}).\label{constraint:2}
\end{align}
We distinguish three cases.
\begin{enumerate}
	\item \label{caseai} Constraint~(\ref{constraint:2}) is not violated for $t= l-1.$ \\
	In this case, Constraint~(\ref{constraint:1}) is violated for $t=l-1$. As a result,
	\begin{align*}
		d(p,x_i) > d(x_i, x_{i+2^{l-1}}) + \lambda d(x_{i+2^{l-1}}, x_{i+2^{l}}). 
	\end{align*} Therefore, 
	\begin{align*}
		d(p,x_{i_{\max}}) &\geq d(p,x_i) - d(x_i, x_{i_{\max}}) 
		\\ & \geq d(x_i, x_{i+2^{l-1}}) + \lambda d(x_{i+2^{l-1}}, x_{i+2^{l}}) - d(x_i, x_{i_{\max}})
		\\ &  \geq d(x_{i_{\max}}, x_{i+2^{l-1}})+ \lambda d(x_{i+2^{l-1}}, x_{i+2^{l}})
		\\ & \geq \lambda d(x_{i_{\max}}, x_{j}) \geq \frac{  1-f }{f} \lambda d(x_i,x_{i_{\max}}) \geq g d(x_i, x_{{\max}}).
	\end{align*} 
	\item \label{caseaii}Constraint~(\ref{constraint:1}) is not violated for $t=t'-1$.\\
	In this case, Constraint~(\ref{constraint:2}) is violated for $t=t'-1$. Therefore, 
	\begin{align*}
		d(p,x_{i+2^{t'}}) > d(x_{i+2^{t'}}, x_{i+2^{t'-1}}) + \lambda d(x_{i},x_{i+2^{t'-1}}).
	\end{align*} As a result,
	\begin{align*}
d(p,x_{i_{\max}}) & \geq 	d(p,x_{i+2^{t'}}) -d (x_{i_{\max}},x_{i+2^{t'}})
\\ & > d(x_{i_{\max}}, x_{i+2^{t'-1}}) + \lambda d(x_{i},x_{i+2^{t'-1}})
\\ & > \lambda d(x_i, x_{i_{\max}}) \geq \frac{f}{1-f}g d(x_i, x_{{\max}})
	\end{align*}
	where $\frac{f}{1-f} > 1$. We keep the constant $\frac{f}{1-f}$ to simplify the proof for case~\ref{caseb}.
\item \label{caseaiii}Constraint~(\ref{constraint:2}) is violated for $t+1$ and Constraint~(\ref{constraint:1}) is violated for $t$ for some $t=t'-1,t' , \dots ,  l-2$. 
In this case,
\begin{align*} 
	d(p,x_{i+2^{t+2}}) & > d(x_{i+2^{t+2}}, x_{i+2^{t+1}}) + \lambda d(x_{i},x_{i+2^{t+1}}) \text{ and }
\\ d(p,x_{i}) & > d(x_{i}, x_{i+2^t}) + \lambda d(x_{i+2^t},x_{i+2^{t+1}}).
	\end{align*} Therefore,
	\begin{align*}
		2d(p,x_{i_{\max}}) & \geq d(p,x_{i+2^{t+2}}) -d(x_{i_{\max}},x_{i+2^{t+2}})+ d(p,x_{i}) -d(x_{i_{\max}},x_{i}) 
		\\ &> \lambda d(x_i,x_{i+2^{t}})+(2\lambda -1 )d(x_{i+2^{t}}, x_{i+2^{t+1}}) >(2\lambda -1)d(x_i, x_{i + 2^{t+1}})
		\\ & \geq (2\lambda -1 ) d(x_i, x_{i_{\max}}) \geq 2\frac{f}{1-f}g d(x_i, x_{{\max}}) 
	\end{align*}
	where $\frac{f}{1-f} > 1$. We keep the constant $\frac{f}{1-f}$ to simplify the proof for case~\ref{caseb}.
\end{enumerate}
	\item \label{caseb}$ i+{2^{l-1}} < i_{\max}\leq j-1 $. \\
	In this case, there exists a positive integer $t' \leq l-1$ such that $j-2^{t'}<i_{\max}\leq j-2^{t'-1}$.
According to the definition of $M^{LEV}_i$, $p$ is not in $\tau_{j-2^{t}}$ for any $t=t'-1,t',\dots, l-1$. 
Note that we did not use $p\in \tau_i$ in the proof of case~\ref{casea} where $i_{\max} \leq i+2^{l-1}$. Therefore, the cases~\ref{casea} and~\ref{caseb} are symmetric to each other (with $i$ and $j$ swapped). In the case symmetric to case \ref{caseai}, we have 
$$d(p,x_{i_{\max}}) \geq \lambda d(x_{i_{\max}}, x_i) \geq g d(x_i, x_{i_{\max}}).$$
For the cases symmetric to cases~\ref{caseaii} and~\ref{caseaiii}, we have
$$d(p,x_{i_{\max}}) \geq \frac{f}{1-f}g d(x_j,x_{i_{\max}}) \geq g d(x_i, x_{i_{\max}}).$$
\end{enumerate}
Since we have covered all possible cases, the proof is completed.
\qedhere
\end{proof}
\subsection{Proof of Lemma~\ref{lemma:type1-to-non-oscillating}}\label{proof:lemma:type1-to-non-oscillating}
\begin{proof}
We define the route of each server $i$ in $NOS$ separately, and prove that for each server $i$, the non-isolating route that we defined is at most four times of the route defined in the given algorithm $RES$. 

To describe the non-oscillating route for server $i$, we first denote the sequence of the regions visited by server $i$ to be $(\tau_{i_1}, \tau_{i_2}, \dots , \tau_{i_q})$ where $q$ is the length of the sequence. According to this definition, $i_1=i_q=i$, and $i_{j}\neq i_{j+1}$ for any $j=1,2,\dots , q-1$.  

For each $j=1,2,\dots , q$, denote $a_j$ and $b_j$ the first and last points (each point could be a request or the depot $x_i$) corresponding to the region $\tau_{i_j}$ visited by server $i$ under the algorithm $RES$. 
Note that according to the definition, $a_1= b_q=x_i$. 

Now we consider the case where the route is oscillating between the $t^{th}$ region and the $t'^{th}$ region and $t'-t\geq 3$, i.e., $i_j=i_{j+2}$ for $j=t, \dots , t'-1$, but $i_j\neq i_{j+2}$ for $j=t-1,t'$. For the case where $t'-t$ is odd, we will define an alternative routes that starts at $a_t$, ends at $b_{t'}$, and does not oscillate at all. For the case where $t'-t$ is even, we adopt the alternative route for the sequence $(\tau_{i_t}, \tau_{i_{t+1}}, \dots , \tau_{i_{t'-1}})$, and then the new route oscillates only in three regions starting at point $a_t$ and ends at point $b_{t'}$.

Let us now define an alternative route for the case where $t'-t$ is odd. To simplify the notation, we denote $r \rightarrow r'$ to be the route described in $RES$ that travels between the two requests (or between a request and a depot) $r$ and $r'$, and $r \dashrightarrow r'$ to be the shortest path to travel from location $r$ to location $r'$.
		In $NOS$, the alternative path of server $i$ follows the following three steps.
		\begin{enumerate}
			\item Travel through the requests in $\tau_{i_t}, \tau_{i_{t+2}}, \dots , \tau_{i_{t'-1}}$ through the order that is the same as $RES$, i.e.,
		$$a_t\rightarrow b_t \dashrightarrow a_{t+2} \rightarrow b_{t+2}\dashrightarrow a_{t+4} \rightarrow b_{t+4} \dots \dashrightarrow a_{t'-1} \rightarrow b_{t'-1}.$$
		\item Go to location $a_{t+1}$, i.e.,
		$$ b_{t'-1}\dashrightarrow a_{t+1}. $$ 
		\item Travel through the requests in $\tau_{i_{t+1}}, \tau_{i_{t+3}}, \dots , \tau_{i_{t'}}$ through the order that is the same as $RES$, i.e.,
		$$a_{t+1}\rightarrow b_{t+1} \dashrightarrow a_{t+3} \rightarrow b_{t+3}\dashrightarrow a_{t+5} \rightarrow b_{t+5} \dots \dashrightarrow a_{t'} \rightarrow b_{t'}.$$
		\end{enumerate}
It is clear that the alternative route does not oscillate. Note that the solid arrows are also in the original route of $RES$ and they do not duplicate. Therefore, the solid arrows do not increase the length of the route. 
The length of the dashed arrows in each of the three steps is smaller than the length of the route that travels from $b_t$ to $a_{t'}$ in the route of $RES$.
 In addition, the route between $b_t$ and $a_{t'}$ does not intersect with different parts of the route that oscillate, even for other servers.  
Therefore, the total additional length due to all dashed arrows for all servers is at most $3RES(I)$. 
As a result, we conclude that $NOS(I) \leq 4RES(I).$
	\qedhere
\end{proof}
\subsection{Proof of Lemma~\ref{lemma:non-oscillating-to-alg}}\label{proof:lemma:non-oscillating-to-alg}
\begin{proof}
Because $$\sum_{i=1}^m DIS^{LEV}(S_i) \geq DIS^{LEV}(\bigcup_{i=1}^m  S_i) $$ for any sets of requests $S_1,\dots, S_m$, it is sufficient to prove the following inequality.
$$ DIS^{LEV}(S^{NOS}_i) \leq \frac{25-5\lambda-6\lambda^2}{1-\lambda} (\text{length of the route of server }i \text{ in }NOS)$$ 
for each $i=1,2,\dots , m$.

Given the route of server $i$ in $NOS$, we will create an algorithm $LEV'$ such that any request in $S^{NOS}_i \cap M^{LEV}_j$ is assigned to server $j$ where $S^{NOS}_i$, as defined in Section~\ref{sec:offline_problem}, is the set of the requests that are assigned to server $i$ under the algorithm $NOS$. 

After describing the algorithm $LEV'$, we will prove the following claim. 
\begin{align} LEV'(S^{NOS}_i) \leq \frac{25-5\lambda-6\lambda^2}{1-\lambda} (\text{length of the route of server }i \text{ in }NOS). \label{inequality:lev'}\end{align}
If Claim~(\ref{inequality:lev'}) holds, then the lemma holds because each server in $DIS^{LEV}(S^{NOS}_i)$ travels through the optimal TSP tour, and thus the cost $DIS^{LEV}(S^{NOS}_i)$ is not greater than $LEV'(S^{NOS}_i)$. 

Let us now describe how we create the algorithm $LEV'$.
Given the route of server $i$ in $NOS$, we define $q$ and $\{ a_j, b_j, i_j \}_{j=1}^q$ in the same way as we defined in the proof of Lemma~\ref{lemma:type1-to-non-oscillating}. 
 If $q=1$, set $LEV'=NOS$, then we are done. Therefore, we assume $q>1$.
  We first note that when $q>1$, $q \geq 3$ because $i_1=i_q=i$ but $i_1\neq i_2$. 
  For $j=1, q$, we add segment $(b_1,x_i)$ and $(x_i, a_q)$ respectively. 
  For each $j=2,3,\dots , q-1$, we add two segments $(x_{i_j},a_j)$ and $(b_j, x_{i_j})$. 
  By doing so, each server $j\in \mathbb{N}_l$ can follow a route that begins and ends at the depot $x_{j}$ and visit all requests in $I\cap M^{LEV}_{j}$. Hence, we have successfully defined a valid algorithm $LEV'$.
  
We are now ready to prove Claim~(\ref{inequality:lev'}).
 The total length that we added is $$ d(b_1,x_i)+d(x_i,a_q) + \sum_{j=2}^{q-1}d(x_{i_j},a_j)+d(b_j,x_{i_j}).$$ 
 It is sufficient to prove that this quantity is at most $\frac{24-6\lambda-6\lambda^2}{1-\lambda}$ times of the length of the route of server $i$ in $NOS$.
To prove this, we consider each $d(x_{i_j},a_j)$ and $d(x_{i_j},b_j)$ separately. For each $j=3,4,\dots , q-2$, we prove the following claim.
	\begin{align}
		d(x_{i_j},a_j)\leq \frac{4-\lambda-\lambda^2}{1-\lambda} (d(a_{j-1}, a_j)+d(a_j,a_{j+1})+d(a_{j+1},a_{j+2})). \label{inequality:each_a_j}
	\end{align} We skip the proof of the cases where $j=1,2,q-1,q$ but similar results hold for those cases. By symmetry, similar results hold when each of the $a$ appeared in Claim~(\ref{inequality:each_a_j}) is replaced with $b$. 
If we have Claim~(\ref{inequality:each_a_j}), we can prove the lemma by summing $d(x_{i_j},a_j)+d(x_{i_j},b_j)$ over all $j$.

	Before proving Claim~(\ref{inequality:each_a_j}), we first prove the following two useful inequalities. Given three integers $t_1, t_2, t_3$ in $\mathbb{N}_l$ such that $t_1<t_2<t_3$ and three points $p_1\in \tau_{t_1}$, $p_2\in \tau_{t_2}$, and $p_3\in \tau_{t_3}$, we have
	\begin{align}
		d(x_{t_2}, p_2) &\leq \frac{2+\lambda}{2}d(p_1,p_3) \label{inequality:middle_point}\\ 
		d(x_{t_3}, p_3) &\leq \frac{4-\lambda-\lambda^2}{2-2\lambda} (d(p_1,p_3)+d(p_2,p_3))\label{inequality:right_point} 
	\end{align} 	 	 
Let us first prove Inequality~(\ref{inequality:middle_point}). To prove this, we first find the following upper bound for $d(x_{t_2}, p_2)$.
\begin{align*}
	2 d(x_{t_2}, p_2) & \leq d(x_{t_2+2^{l}}, p_2)+d(x_{t_2+2^{l}}, x_{t_2}) +d(x_{t_2-2^{l}}, p_2)+d(x_{t_2-2^{l}}, x_{t_2})
	\\ & \leq (2+\lambda) d(x_{t_2+2^{l}},x_{t_2-2^{l}}).
		\end{align*}
	Then, we find the following lower bound for $d(p_1,p_3)$.
	\begin{align}
		d(p_1,p_3) & \geq d(x_{t_1-2^{l}}, x_{t_3+2^{l}}) - d(x_{t_1-2^{l}}, p_1) -d(x_{t_3+2^{l}}, p_3) \nonumber
		\\ & > d(x_{t_1+2^{l}}, x_{t_3-2^{l}}) + (1-\lambda)d(x_{t_1+2^{l}},x_{t_1})+(1-\lambda)d(x_{t_3-2^{l}},x_{t_3}) \nonumber
		\\ & > d(x_{t_2+2^{l}},x_{t_2-2^{l}}) \label{inequality:distance_p_1_p_3}.
		\end{align}
		Combine the two inequalities above, we obtain Inequality~(\ref{inequality:middle_point}).
		
		Inequality~(\ref{inequality:right_point}) is a direct result of the following inequality.
		\begin{align*}
			d(x_{t_3},p_3)&\leq d(x_{t_3},x_{t_2})+d(x_{t_2},p_2)+d(p_2,p_3)
			\\ & \leq \frac{1}{1-\lambda}d(p_1,p_3)+\frac{2+\lambda}{2}d(p_1,p_3)+d(p_2,p_3)
		\end{align*}
where the last inequality follows from Inequalities~(\ref{inequality:middle_point}) and~(\ref{inequality:distance_p_1_p_3}).

We are now ready to prove Claim~(\ref{inequality:each_a_j}). Without loss of generality, we assume that $i_{j-1}<{i_j}$. We distinguish three cases.
\begin{enumerate}
\item $i_{j+1}<i_{j}$ and $i_{j-1}\neq i_{j+1}$. In this case, according to Inequality~(\ref{inequality:right_point}), 
$$ d(x_{i_j},a_j)\leq \frac{4-\lambda-\lambda^2}{2-2\lambda}  (d(a_{j-1},a_j)+d(a_j,a_{j+1})).$$
\item $i_{j}<i_{j+1}$. In this case, according to Inequality~(\ref{inequality:middle_point}), 
$$d(x_{i_j},a_j)\leq \frac{2+\lambda}{2}d(a_{j-1},a_{j+1})\leq \frac{2+\lambda}{2}(d(a_{j-1},a_j)+d(a_j,a_{j+1}) ).$$
\item $i_{j-1}=i_{j+1}$. In this case, we consider $i_{j+2}$. Because the route is non-oscillating, $i_{j+2}\neq {i_j}$. We further distinguish two cases.
\begin{enumerate}
	\item $i_{j+2} < {i_j}$. In this case, according to Inequality~(\ref{inequality:right_point}), 
\begin{align*} d(x_{i_j},a_j) & \leq \frac{4-\lambda-\lambda^2}{2-2\lambda}  (d(a_{j},a_{j+2})+d(a_j,a_{j+1})) 
	\\ & \leq\frac{4-\lambda-\lambda^2}{1-\lambda}  (d(a_{j+1},a_{j+2})+d(a_j,a_{j+1})).	
\end{align*}
	\item $i_{j+2} > {i_j}$. In this case, according to Inequality~(\ref{inequality:middle_point}), 
$$d(x_{i_j},a_j) \leq \frac{2+\lambda}{2}d(a_{j+1},a_{j+2}).$$
\end{enumerate}
	\end{enumerate} 
	Hence the proof is completed.
	\qedhere
\end{proof}
\subsection{Proof of Theorem~\ref{thm:local_partition}}\label{proof:thm:local_partition}
\begin{proof}
	Without loss of generality, let us assume that the minimum distance between a pair of depots is one, i.e., $\min_{i,j}\{d(x_i,x_j)\}=1$. 

Let us first prove that the approximation ratio given by the partition scheme $LOC$ is $\Omega(f(m))$. Consider the following example. Let the metric space be $\mathbb{R}$, $m=3$ and $x_1=0$, $x_2=1$, $x_3=f(m)+1$. Let there be only one request and the request is located at $1+1/4$. In this case, the request will be assigned to server $3$, and the cost of the distributed algorithm is $DIS^{LOC}(I)=2(f(m)-1/4)$. On the other hand, the optimal algorithm can assign the request to server $2$ and the cost would be $OPT(I)=1/2$. As a result, the ratio $DIS^{LOC}(I)/OPT(I)=4f(m)-1$, which is $\Theta(f(m))$. Therefore, the partition scheme $LOC$ has an approximation ratio of $\Omega(f(m))$. 

Let us prove that the partition scheme $LOC$ leads to an approximation ratio of $O(f(m))$ by showing that it is at most $2+4f(m)$. 

We first divide the requests into two sets $L_1=I\cap \bigcup_{i=1}^{m-1}M_i$, and $L_2=I\cap M_m$. 
We have $DIS^{LOC}(I) = DIS^{LOC}(L_1) + DIS^{LOC}(L_2)$, $OPT(I) \geq OPT(L_1)$, and $OPT(I) \geq OPT(L_2)$. 
Therefore, the following two claims are sufficient to proving the theorem.
\begin{align}
	DIS^{LOC}(L_1)=OPT(L_1) \label{equation:L_1}
\end{align}
and 
\begin{align}
	DIS^{LOC}(L_2)\leq (1+4f(m)) OPT(L_2). \label{inequality:L_2}
\end{align}
We first prove Claim~(\ref{equation:L_1}). We prove this by showing that, when $I=L_1$, in the optimal algorithm $OPT$, $S^{OPT}_i$ consists only requests in $M^{LOC}_i$ for any $i=1,2,\dots m-1$. Using the same argument for proving this, it will be evident that $S^{OPT}_m$ is empty. To prove this, we first note that the distance between two points in different regions is greater than $1/2$. It is because if $p_1 \in M^{LOC}_i$, $p_2 \in M^{LOC}_j$, and $i\neq j$, then $$ d(p_1,p_2) \geq d(x_i,x_j)-d(x_i,p_1)-d(x_j,p_2) > 1-1/4-1/4 =1/2.$$
Now let us assume on the contrary that server $i$ serves at least one request not in $M^{LOC}_i$ in the optimal algorithm. Assume that after leaving $M^{LOC}_i$, server $i$ visits requests in regions in the order of $M^{LOC}_{i_1}$, $M^{LOC}_{i_2}$, $\dots$, and $M^{LOC}_{i_l}$ and then back to $M^{LOC}_i$. 
For each $l'=1,2,\dots , l$, denote $a_{i_{l'}}$ and $b_{i_{l'}}$ the first and last requests visited in the region $M^{LOC}_{i_{l'}}$ ($a_{i_{l'}}=b_{i_{l'}}$ if there is only one such request). 
To simplify the notation, we denote $b_{i_0}$ the last point (request or depot) server $i$ visits before leaving $M^{LOC}_i$ and $a_{i_{l+1}}$ the first point (request or depot) server $i$ visits after returning $M^{LOC}_i$. 

Now we consider alternative routes of servers by removing the $(l+1)$ edges $\{(b_{i_{l'}}, a_{i_{l'+1}})\}_{l'=0}^l$ and adding the $(2l+2)$ edges $(x_i,b_{i_0})$, $(x_i, a_{i_{l+1}})$, and $\{(x_{i_{l'}},a_{i_{l'}}), (x_{i_{l'}},b_{i_{l'}})\}_{l'=1}^l$.

The length of each removed edge is greater than $1/2$, and the length of each added edge is at most $1/4$. Therefore, the length of the new solution is smaller than that of the optimal solution, which is a contradiction. Hence, Claim~\ref{equation:L_1} is true.

Let us now prove Claim~(\ref{inequality:L_2}). The inequality obviously holds if $S^{OPT}_i$ is empty for all $i=1,2,\dots , m-1$. If $S^{OPT}_i$ is not empty for any $i=1,2,\dots , m-1$, let $a_i$ and $b_i$ be the first and last such request in the route of server $i$ in the optimal solution ($a_i=b_i$ if there is only one such request). For each $i=1,2,\dots m-1$, we replace $(a_i,x_i)$ and $(b_i,x_i)$ with $(a_i, x_m)$ and $(b_i, x_m)$ such that all requests in $S^{OPT}_i$ are served by server $m$ in the new solution. Clearly, the length of the new solution is an upper bound of $TSP_m(L_2)$, which is the same as $DIS^{LOC}(L_2)$.

We note that for any $p\in L_2$, $\frac{d (p,x_m)}{d(p,x_i)} <1+4f(m)$ because
$$ \frac{d (p,x_m)}{d(p,x_i)} \leq \frac{d(p,x_i) + d(x_i,x_m)}{d(p,x_i)} < 1+4f(m).$$
Therefore, the length of the new solution is at most $(1+4f(m))OPT(L_2)$. Thus, Claim~(\ref{inequality:L_2}) holds.

\end{proof}


\end{document}